\documentclass[12pt,draftcls,onecolumn]{IEEEtran}
\usepackage{amsmath,amssymb,amsthm}    
\usepackage[dvips]{graphicx}   
\usepackage{verbatim}   
\usepackage{color}      
\usepackage{subfigure}  
\usepackage{epsfig}
\usepackage{float}
\usepackage{algorithm}
\usepackage{algorithmic}
\usepackage{setspace}
\usepackage{subfigure}
\usepackage{hyperref} 
\usepackage{amssymb}
\newtheorem{lemma}{\bf Lemma}
\newtheorem{theorem}{\bf Theorem}
\newtheorem{corollary}{\bf Corollary}

\DeclareMathOperator*{\argmin}{arg\,min}

\makeatletter
\newcommand*{\rom}[1]{\expandafter\@slowromancap\romannumeral #1@}
\makeatother

\newcommand{\vx}{{\bf x}}
\newcommand{\vy}{{\bf y}}
\newcommand{\vw}{{\bf w}}

\newcommand{\ms}{{\bf S}}
\newcommand{\mr}{{\bf R}}
\newcommand{\mi}{{\bf I}}
\newcommand{\mc}{{\bf C}}
\newcommand{\md}{{\bf D}}
\newcommand{\mt}{{\bf T}}
\newcommand{\ma}{{\bf A}}
\newcommand{\mb}{{\bf B}}
\newcommand{\mz}{{\bf Z}}
\newcommand{\mg}{{\bf G}}

\newcommand{\tmr}{\widetilde{\bf R}}
\DeclareMathOperator{\rank}{rk}
\DeclareMathOperator{\spn}{span}
\DeclareMathOperator{\dmn}{dim}

\begin{document}

\title{\vspace{-0.7cm}Subspace-Aware Index Codes}
\author{ 
Bhavya~Kailkhura$^{*\dagger}$,~\IEEEmembership{Member,~IEEE}, Lakshmi~Narasimhan~Theagarajan$^{*\ddagger}$,~\IEEEmembership{Member,~IEEE},
Pramod~K.~Varshney$^\ddagger$,~\IEEEmembership{Fellow,~IEEE}
\thanks{This work was supported in part by NSF Grant no. ECCS 1609916.}
\thanks{This work was performed under the auspices of the U.S. Department of Energy by Lawrence Livermore National Laboratory under Contract DE-AC52-07NA27344. LLNL-JRNL-718227}
\thanks{*These authors contributed equally to this work.}
\thanks{$\dagger$ This author is presently affiliated to Lawrence Livermore National Laboratory, kailkhura1{@}llnl.gov}
\thanks{$\ddagger$These authors are at the Department of EECS, Syracuse University, New York, \{ltheagar, varshney\}{@}syr.edu.}
\vspace{-1cm}
}
\date{}
\maketitle

\begin{abstract} 
In this paper, we generalize the well-known index coding problem to exploit 
the structure in the source-data to improve system throughput. In many 
applications (e.g., multimedia), the data to be transmitted may lie (or can be well approximated) in a low-dimensional 
subspace. We exploit this low-dimensional structure of the data using an algebraic 
framework to solve the index coding problem (referred to as {\em 
subspace-aware index coding}) as opposed to the traditional index coding 
problem which is {\em subspace-unaware}. Also, we propose an efficient 
algorithm based on the alternating minimization approach to obtain near 
optimal index codes for both subspace-aware and -unaware cases. Our 
simulations indicate that under certain conditions, a significant throughput 
gain (about $90\%$) can be achieved by subspace-aware index codes over 
conventional subspace-unaware index codes.
\end{abstract}

\vspace{-2mm}

\vspace{-4mm}
\section{Introduction}
\label{sec1}

\begin{figure*}[t]
    \centering
    \subfigure[A set of eigenfaces or eigenvectors.]{
        \includegraphics[width=0.4\textwidth]{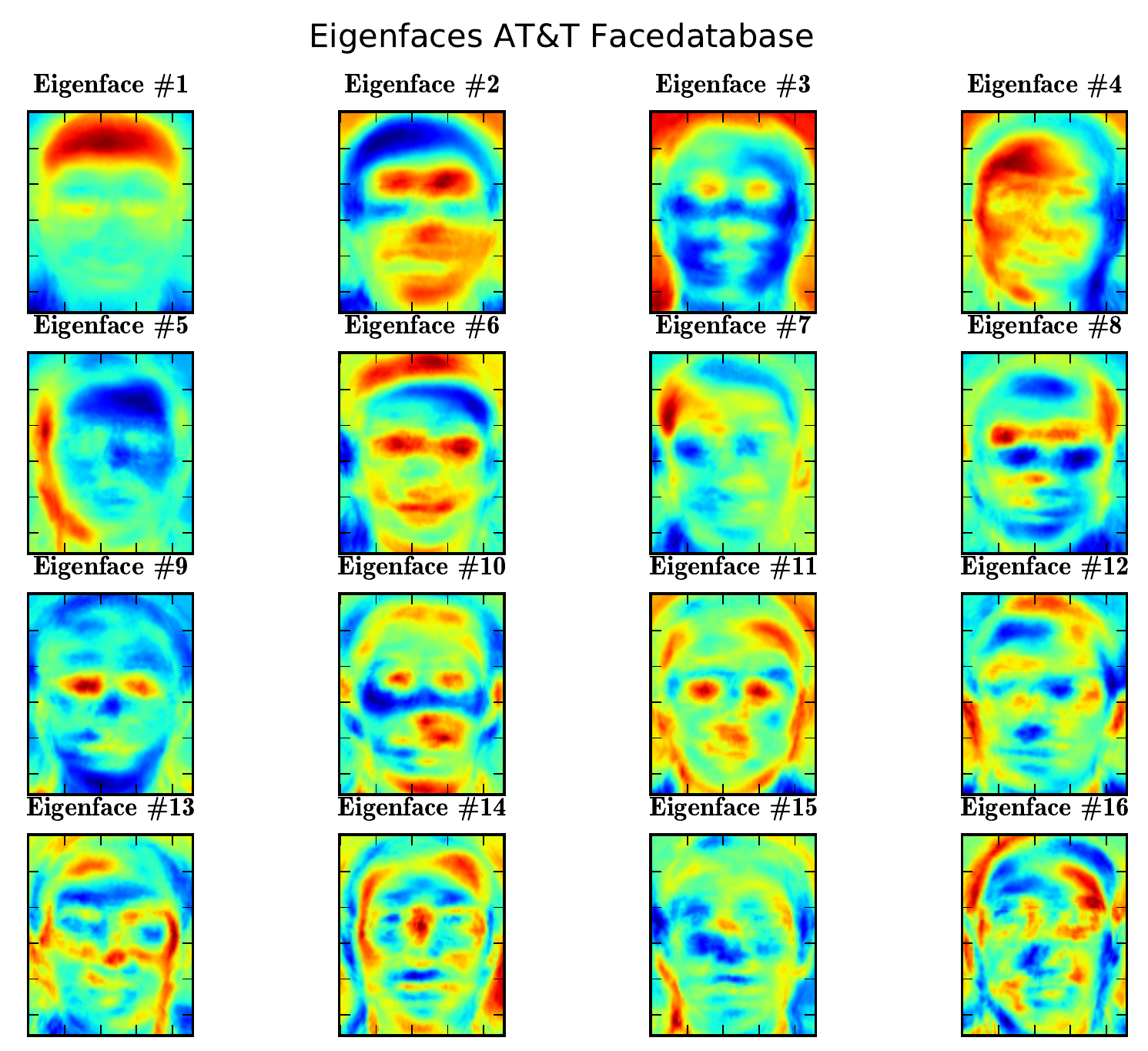}
        \label{fig:eigenface} } \hspace{4mm}
    \subfigure[A face can be approximated as a linear combination of the best $K$ eigenvectors.]{
        \includegraphics[width=0.45\textwidth]{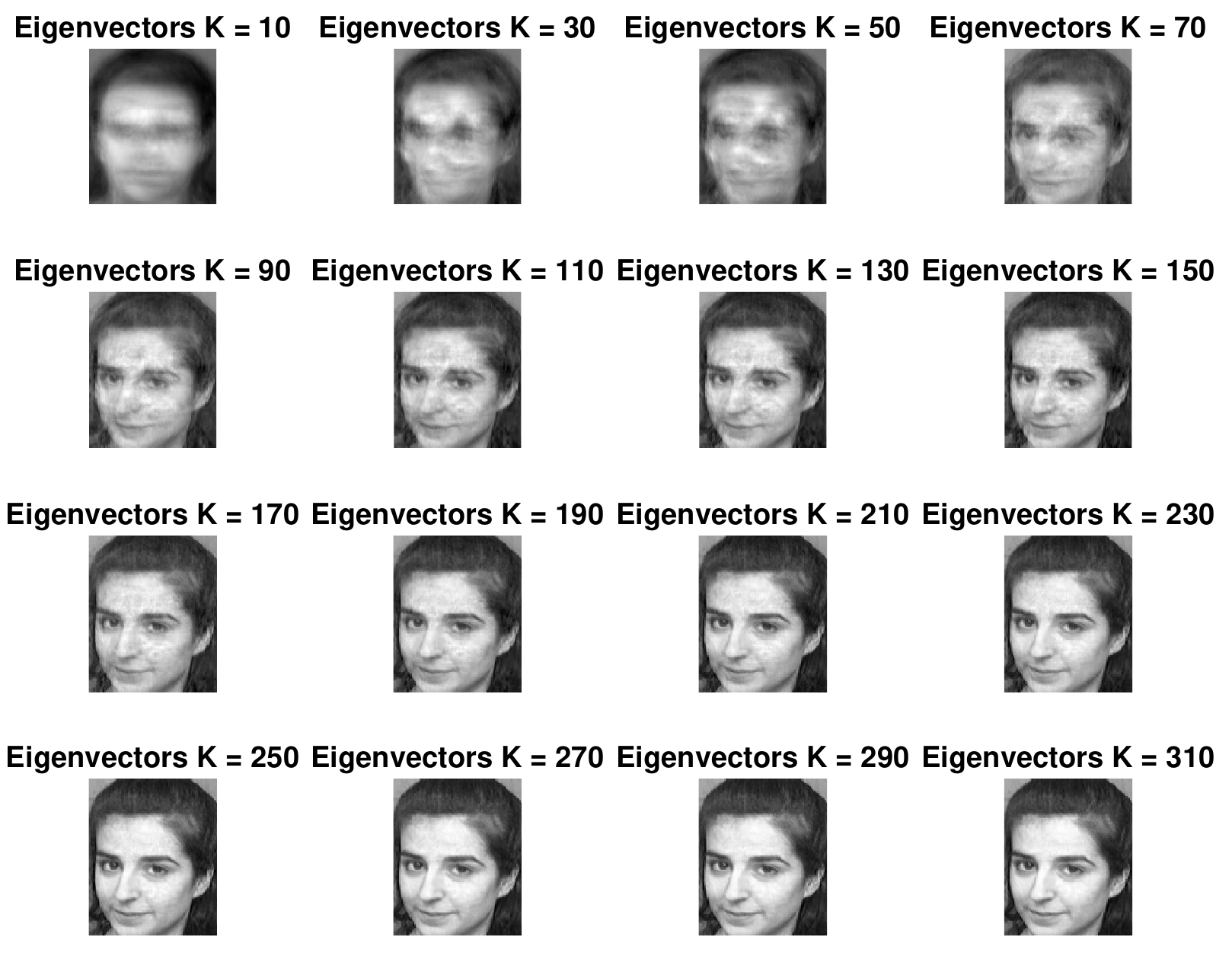}
        \label{fig:recon} }        
    \caption{Illustration of practical scenarios where data lies in a low dimensional linear subspace.}
\end{figure*}

Index coding with side-information (ICSI)~\cite{Birk,Bar, cache}, is a problem, where 
a server has $N$ stored messages that it can broadcast over a noiseless channel 
to a set of receivers or clients. Each client has a subset of the $N$ messages 
as side information, and requests a subset of messages that it needs from the 
server. The objective of the ICSI problem is to devise an optimal coding strategy 
that minimizes the number of broadcast transmissions made by the server to satisfy 
the requirements of all the clients. The optimality criterion of an index code is its 
code length. The optimal index code length, i.e., the minimum number of transmissions 
required from the server for successful recovery of the desired information at the 
clients, was first characterized in \cite{Birk, Bar} as the minimum rank of a matrix 
that represents the side-information graph~\cite{Bar}. A method to construct index 
codes by solving a matrix completion problem was presented in \cite{Tan}. In 
many practical scenarios, users possess coded side information (CSI) \cite{Shum}; and 
index codes for linear CSI were studied in \cite{Shum, dimakis}. Index codes over real 
field and their construction methods are investigated in \cite{iitc, real2}. 
Considering index codes over the real field enables the use of efficient optimization
techniques to construct near optimal index codes.
Further, it was shown in \cite{Effros} that network codes are equivalent to index codes. 
Thus, one can construct an optimal network code by constructing an optimal index code 
for the equivalent problem. Network codes over real field are discussed in 
\cite{ncreal1, ncreal2}. 

{The source-data encountered in many practical systems such as data caching, 
images, video streaming, big-data storage and processing, can be well approximated 
using a lower dimensional linear subspace~\cite{ss1, ss2, ss3} (and references therein). 
Motivated by such applications, we propose a technique to construct index codes when 
the source-data belongs to a lower dimensional linear subspace.} For example, it is well known that a facial image is
a point from a high-dimensional image space which can be well approximated in a lower-dimensional linear subspace. The lower-dimensional subspace is found using Principal Component Analysis, which identifies the axes with maximum variance. In Figure~\ref{fig:eigenface}, a set of eigenvectors (known as eigenfaces) are shown for AT\&T Facedatabase. There are ten different images of each of the 40 distinct subjects.  The size of each image is 92x112 pixels, with 256 grey levels per pixel. In Figure~\ref{fig:recon}, we can see that a good reconstruction quality can be obtained using a very small number of eigenfaces.\footnote{These results are obtained from \text{https://github.com/bytefish/facerec}.}

More specifically, 
the main contributions of this paper can be summarized as follows.
\begin{itemize}
\item We generalize the index coding problem with coded (and/or uncoded) side 
information to exploit the low-dimensional structure that may be present in the source-data. 
\item We establish bounds on the gain achieved by subspace-aware index codes over 
subspace-unaware case.
\item We consider the design of subspace-aware/unaware index codes with coded/uncoded 
side information in a unified optimization framework and develop an efficient algorithm 
to construct near optimal index codes.
\item Finally, we provide theoretical guarantees and simulation results on the 
performance of the proposed techniques.
\end{itemize}
The notations followed in the rest of this paper are: $\rank(.)$ denotes the 
rank of a matrix, $\spn(.)$ denotes a vector space spanned by a set of vectors, 
$(.)^\dagger$ denotes the pseudo-inverse of a matrix, and $\|.\|_F$ denotes the 
Frobenius norm of a matrix.

\vspace{-4mm}
\section{Problem setup}
\label{sec2}
\vspace{-2mm}
Consider a network with $U$ users and a data source (DS). Let $N$ denote the 
total number of data packets involved in a transmission instance, $P$ denote 
the size of each data packet, $\vx_i$ denote the data in the $i$th packet, 
$\vx_i\in\mathbb{R}^P$ for $i=1,2,\cdots,N$, and  
$\vx\triangleq[\vx_1, \vx_2, \cdots, \vx_N]^T\in\mathbb{R}^{PN}$. 
The $j$th user requests $V_j$ number of packets from the DS, $\mathcal{R}_j$ 
is the set of all indices of the requested data packets by the $j$th user, 
$|\mathcal{R}_j|=V_j$ for $j=1,2,\cdots,U$, and $\vx_{\mathcal{R}_j}$ denotes
the $PV_j\times 1$ information vector requested by the $j$th user.
Each user possesses a linearly coded side information. Let $M_j$ denote 
the length of the CSI and $\ms_j\in\mathbb{R}^{PM_j\times PN}$ denote the side 
information coding matrix for the user $j$, $0\leq M_j <N$. The CSI of the 
$j$th user is given by the vector $\ms_j\vx$. When the $j$th user has uncoded
side information (USI), the side information consists of $M_j$ data packets, and
the non-zero columns of $\ms_j$ form an identity matrix of dimension $PM_j\times PM_j$.

If the vector $\vx$ belongs to a low-dimensional subspace, then $\vx=\mt\vw$, 
where $\mt\in\mathbb{R}^{PN\times PD}$ ($1\leq D<N$) is the matrix of basis vectors
of the low-dimensional subspace, $\vw\in\mathbb{R}^{PD}$, and $\rank(\mt)=PD$.

\textbf{Goal}: Knowing $\mathcal{R}_j$, matrices $\ms_j$ and subspace structure 
$\mathbf{T}$ for $j=1,2,\cdots,U$, the goal is to have the DS broadcast the least number of coded 
data packets to  $U$ users such that each user is able to successfully decode the 
requested packets. \qedsymbol

Let $\vy\triangleq[\vy_1, \vy_2, \cdots, \vy_L]^T\in\mathbb{R}^{PL}$ be the
data vector transmitted by the DS. Now, each user needs to decode 
$\vx_{\mathcal{R}_j}\in\mathbb{R}^{PV_j}$ from $[(\ms_j\vx)^T\, \vy^T]^T$.
Assuming linear decoding, the $j$th user performs the decoding as 
$\widehat{\vx}_{\mathcal{R}_j}=\md_j[(\ms_j\vx)^T\, \vy^T]^T$, where $\md_j$ is the 
decoding matrix. For linear encoding, this problem can be stated as follows.

\textbf{Problem}: Find a matrix $\mc\in\mathbb{R}^{PL\times PN}$ such that 
\begin{equation}
\md_j\begin{bmatrix}\ms_j \mathbf{Tw}\\ \vy\end{bmatrix} = \vx_{\mathcal{R}_j}, \, \forall j, \mbox{ s.t. } \vy=\mc\vx=\mc\mt\vw.
\label{eqn:cond}
\end{equation}

We refer to the matrix $\mc$ as the $L$-length index code. For a given 
$\{\mathcal{R}_j\}$, $\{\ms_j\}$, and $\mt$, the matrix $\mc$ with the least 
number of rows $L^*$ satisfying the condition in (\ref{eqn:cond}) is the optimal 
index code and $L^*$ is the optimal index code length\footnote{Note that, in our
proposed methodology, compression and index coding are performed in a unified framework. 
This helps to further simplify the receiver by relieving it of the separate decompression 
algorithm, reduces computational complexity, and improves overall system throughput.}. 

\begin{figure*}[t]
    \centering
    \subfigure[A wireless relay network.]{
        \includegraphics[width=0.3\textwidth]{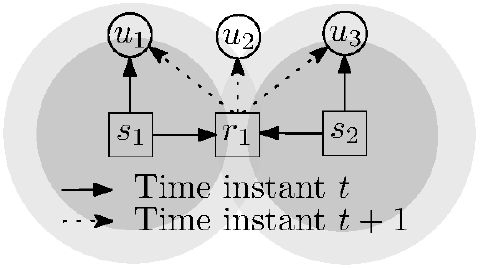}
        \label{fig:relay} } \hspace{4mm}
    \subfigure[A collaborative cognitive radio network.]{
        \includegraphics[width=0.3\textwidth]{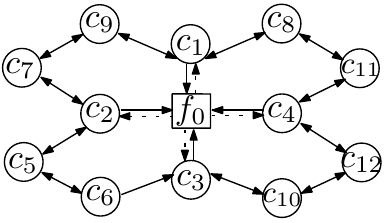}
        \label{fig:sensors} }  \hspace{4mm}
    \subfigure[Users connected to datacenters.]{
        \includegraphics[width=0.25\textwidth, height=0.17\textwidth]{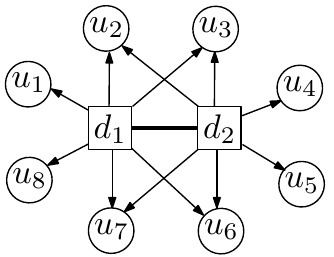}
        \label{fig:datacenter} }        
    \caption{Illustration of practical scenarios where index codes can be employed to improve efficiency and throughput.}
\end{figure*}

\textbf{Examples}: The problem described above is often encountered in practical 
scenarios such as cloud networks, multicast video-streaming and content-sharing. 
Since the users are connected to multiple datacenters, each user may have different 
subsets of the same data and require different subsets. The data could be low-dimensional 
due to its inherent nature (e.g., videos, images, and sensory data \cite{ss1,ss2}) 
or the usage of redundancy-inducing error correcting codes \cite{icml}. Here, 
the datacenters employ index codes to serve the users' requests to increase network 
efficiency and throughput. A similar problem is also encountered in distributed computing 
setups~\cite{Lee}, distributed cognitive radio networks and satellite networks. Next, we describe these applications in more detail. 

($1$) Consider a wireless relay network represented in 
Fig. \ref{fig:relay}. Source nodes $s_1$ and $s_2$ wish to broadcast their 
data to user nodes $u_1$, $u_2$, and $u_3$ with the help of a relay node $r_1$. 
The intensity signals transmitted by $s_1$ and $s_2$ decay with distance and the SNR 
deteriorates to the extent that their signals are not decodable beyond a 
certain radius of transmission. Consequently, in a transmission time slot $t$, 
nodes $u_1$ and $r_1$ successfully decode the data from $s_1$, while nodes 
$u_3$ and $r_1$ successfully decode the data from $s_2$, and the node $u_3$ 
receives a linear combination of the data from $s_1$ and $s_2$. In time slot 
$t+1$, the relay node broadcasts a coded combination of data from $s_1$ and 
$s_2$ with which the nodes $u_1$, $u_2$, and $u_3$ are able to decode the data
from both $s_1$ and $s_2$. Here, the relay node comes up with an index 
code such that $\mathcal{R}_1=\{2\}, \mathcal{R}_2=\{1,2\}, \mathcal{R}_3=\{1\}$,
$\ms_1=[1, 0]$ (USI), $\ms_2=[h_1, h_2]$ (CSI), $\ms_3=[0, 1]$ (USI), where 
$h_1$ and $h_2$ are the linear coefficients at $u_2$. When the transmitted data 
is encoded with the same linear channel code at $s_1$ and $s_2$, the 
transmitted data belongs to a low dimensional subspace. The columns of 
$\mt$ are the bases of this subspace created by the linear channel code.  

($2$) Consider a collaborative cognitive radio (CR) network of mutiple 
low-cost CR $c_i$ where $i=1,2,\cdots,12$, as illustrated in Fig. 
\ref{fig:sensors}. Each CR senses a disjoint band of a wide spectrum and 
broadcasts the information to all its nearest neighbors. The data collected 
by all the CRs are finally fused at the fusion centers (FC) $f_0$. The FC forms 
the complete map of the wideband spectrum. This complete map has to be 
conveyed back to the CRs. Since the CR network is power constrained, the FC 
conveys this information to its neighbors in least number of broadcasts using 
an index code. Further, as the low-cost CRs have limited memory, each CR stores 
only a linear combination of all the data it receives. Further, the CRs 
$c_1, c_2, c_3$, and $c_4$ develop index codes to broadcast to their neighbors. 
At the FC, $\mathcal{R}_1=\{1, 8, 9\}^c$, $\mathcal{R}_2=\{2, 5, 7\}^c$, 
$\mathcal{R}_3=\{3, 6, 10\}^c$, $\mathcal{R}_4=\{4, 11, 12\}^c$,
\[ \ms_{1[1,8,9]}=\begin{bmatrix}1&0&0\\0&h_8&h_9\end{bmatrix}\hspace{-1mm}, 
\ms_{2[2,5,7]}=\begin{bmatrix}1&0&0\\0&h_5&h_7\end{bmatrix}\hspace{-1mm}, \]
\[ \ms_{3[3,6,10]}=\begin{bmatrix}1&0&0\\0&h_6&h_{10}\end{bmatrix}\hspace{-1mm}, 
\text{ and } \ms_{4[4,11,12]}=\begin{bmatrix}1&0&0\\0&h_{11}&h_{12}\end{bmatrix}
\hspace{-1mm}.\] Due to the inherent sparsity in the wideband spectrum 
activity, the spectrum data is low dimensional in nature \cite{spectrum}. 

($3$) Consider a network of devices served by a central server of facial image
databases, with each device requesting a few facial images while possessing images
of other faces. Such a scenario commonly occurs in biometric verification systems and 
security monitoring applications. As described in the previous section, it is 
known that the facial image data of hundreds of pixels in dimension belong to a 
lower dimensional linear subspace \cite{sbs1}. Therefore, a subspace-aware index
coding in this scenario will improve throughput, speed and scalability of the 
network.

($4$) Consider a network of users connected to a cloud of datacenters hosting 
a common dataset. One such cloud network in illustrated in Fig. \ref{fig:datacenter}. 
The online users could be simultaneously performing operations such as 
document-editing or video-streaming or file-sharing. Since the users are 
connected to multiple datacenters, each user may contain different subsets of 
the same data and require different subsets. The data could belong to a 
low-dimensional linear subspace due to either its inherent nature (e.g., videos
\cite{sbs2}) or the usage of redundancy-inducing error correcting codes. For 
multicast transmissions, the data centers employ index codes to serve the users' 
requests. This increases the overall network efficiency and throughput.

\vspace{-2mm}
\section{Optimal Index Code length}
\vspace{-1mm}
An important step towards solving the problem stated in Sec. \ref{sec2} is to 
identify the minimum length of the index code. Without loss of generality, we 
assume $P=1$. Let $\mr_j$ be a $V_j\times N$ matrix such that 
$\mr_{j}\vx=\vx_{\mathcal{R}_j}$. 
Splitting $\md_j$ into sub-matrices $\ma_j\in\mathbb{R}^{V_j\times M_j}$ 
and $\mb_j\in\mathbb{R}^{V_j\times L}$, we can write (\ref{eqn:cond}) as
\vspace{-3mm}
\begin{eqnarray}
\md_j\begin{bmatrix}\ms_j\mt\vw\\ \vy\end{bmatrix} &\hspace{-3mm}=& \hspace{-3mm}
\begin{bmatrix}\ma_j&\mb_j\end{bmatrix}\begin{bmatrix}\ms_j\mt\vw\\ \vy\end{bmatrix}
=(\ma_j\ms_j + \mb_j\mc)\mt\vw\nonumber\\
\vx_{\mathcal{R}_j}&\hspace{-3mm}=&\hspace{-2mm}\mr_j\vx, \quad \forall j. \label{eqn:base}
\end{eqnarray}

\vspace{-1mm}
Since $\vx=\mt\vw$ and $\vw$ can be any arbitrary vector in $\mathbb{R}^{D}$,
from (\ref{eqn:cond}) and (\ref{eqn:base}), we can write 
$\mb_j\mc\mt=(\mr_j-\ma_j\ms_j)\mt$, $\forall j$. This can be expressed 
succinctly as 
\begin{equation}
\mb\mc\mt=(\mr-\ma\ms)\mt=\widetilde{\mr}\mt,
\label{eqn:system}
\end{equation}
where 
\vspace{-6.25mm}
\begin{eqnarray}
\mb&\triangleq&[\mb_1^T,\mb_2^T,\cdots,\mb_U^T]^T\in\mathbb{R}^{(\sum_jV_j)\times L},\nonumber\\
\ms&\triangleq&[\ms_1^T,\ms_2^T,\cdots,\ms_U^T]^T\in\mathbb{R}^{(\sum_jM_j)\times N},\nonumber\\
\ma&\triangleq&\mathrm{diag}([\ma_1, \ma_2, \cdots, \ma_U]),\nonumber\\
\mr&\triangleq&[\mr_1^T,\mr_2^T,\cdots,\mr_U^T]^T,\nonumber\\
\widetilde{\mr}&\triangleq&\mr-\ma\ms\in\mathbb{R}^{\sum_jV_j\times D}.\nonumber\\
\md_j&\triangleq&[\ma_j\hspace{3mm} \mb_j]
\end{eqnarray}

\vspace{-1mm}
Now, the optimal index code is the matrix $\mc$ that satisfies (\ref{eqn:system}) 
and has the least value of $L (>0)$. Since $\mc$ has only linearly independent
rows, the rank of $\mc$ is $L$. Therefore, the goal is to minimize $\rank(\mc)$ 
such that (\ref{eqn:system}) is satisfied.

{\em Note}: When index coding is performed without the knowledge of the 
underlying subspace (we refer to this scenario as the \textit{subspace-unaware} 
case) or when the data is not low-dimensional, we have 
$\mt=\mi$.

\vspace{-0mm}
\begin{lemma}
$\rank(\mc\mt)=\rank(\tmr\mt)$.
\label{lemma:1}
\end{lemma}
\begin{proof}
If $\sum_jV_j<L$, then the index code length is larger than the number of data packets 
required. Therefore, $\sum_jV_j\geq L$; hence, $\rank(\mb)\leq L$. As governed by 
(\ref{eqn:base}) and (\ref{eqn:system}), when the decoding is successful at the receivers, 
$\mc\vx\in\spn(\mb)$; hence, $\rank(\mb)\geq \dmn(\mc\vx)=L$.
This proves that $\rank(\mb)=L$.

Note that by choosing the index code as $\mc=(\mt^T\mt)^{-1}\mt^T$ (i.e., $L=D$), and 
the decoder matrices as $\mb=\mr\mt$ and $\ma={\bf 0}$, all the required packets can be 
trivially decoded at the receivers. Therefore, the index code is optimal only when 
$L\leq D$. Now, $\rank(\mc\mt)\leq\min(L,D)=L$, and we have
\begin{equation} 
\rank(\mb\mc\mt)\leq\min(\rank(\mb),\rank(\mc\mt))=\rank(\mc\mt).
\label{eqn:l1}
\end{equation}
Further, by Sylvester's rank inequality, 
\begin{equation} 
\rank(\mb\mc\mt)\geq\rank(\mb)+\rank(\mc\mt)-L=\rank(\mc\mt).
\label{eqn:l2}
\end{equation}
From (\ref{eqn:system}), (\ref{eqn:l1}) and (\ref{eqn:l2}), 
$\rank(\mb\mc\mt)=\rank(\tmr\mt)=\rank(\mc\mt)$.
\end{proof}

\vspace{-1mm}\hspace{-4.5mm}
Now, from Sylvester's rank inequality, we get
\begin{equation} 
\rank(\mc)\leq\rank(\mc\mt)+N-D=\rank(\tmr\mt)+N-D.
\label{eqn:upb}
\end{equation}
Since, $N-D$ is a fixed positive value, minimizing $\rank(\tmr\mt)$ 
minimizes the upperbound on $\rank(\mc)$, thereby reducing $\rank(\mc)$.
We use this approach of minimizing $\rank(\tmr\mt)$ to construct index 
codes for low-dimensional data. Further, when $\mc\mt$ has full row-rank
(i.e., $\rank(\mc\mt)=L$), we have $\rank(\mc\mt)=\rank(\mc)$.
For subspace-unaware case, we have $\rank(\mc)=\rank(\tmr)$ \cite{dimakis}.

\vspace{-3mm}
\subsection{Throughput Gain}
The length of the optimal subspace-aware index codes is defined as the following

\vspace{-2mm}
{\small
\begin{equation}
L^*=\min_{\ma_1, \ma_2, \cdots, \ma_U} \rank(\tmr\mt).
\label{eqn:optim}
\end{equation}}
Next, we characterize the throughput gain obtained using subspace-aware 
index codes.

\vspace{-1mm}
\begin{theorem}
\label{thm1}
The length of the optimal linear index code obtained for the subspace-aware case is  
less than or equal to the length of the optimal linear index code obtained for the 
subspace-unaware case.
\end{theorem}

\vspace{-1mm}
\hspace{-3.5mm}{\em Proof.}
Let $\tilde{L}\triangleq\min_{\ma}\rank(\tmr)$ be the optimal subspace-unaware linear 
index code length, and $\widetilde{\ma}\triangleq\argmin_{\ma}\rank(\tmr)$. Now,

\vspace{-5mm}
\begin{eqnarray}
L^*=\min_{\ma} \rank((\mr-\ma\ms)\mt) &\leq& \rank((\mr-\widetilde{\ma}\ms)\mt) \nonumber\\
& \leq& \min(\rank(\mr-\widetilde{\ma}\ms),\rank(\mt)) \nonumber\\
&\leq& \rank(\mr-\widetilde{\ma}\ms) = \tilde{L}. \quad\qedsymbol \nonumber
\end{eqnarray}

\vspace{-3mm}
\begin{corollary}
\label{boundperf} 
The length of the optimal linear index code obtained in the subspace-aware case can be 
bounded as
\begin{equation*}
\min(\tilde{L}-(N-D),1) \leq L^* \leq \tilde{L}
\end{equation*}
\end{corollary}
\begin{proof}
By Sylvester's rank inequality, for any matrix $\ma$,
\begin{equation}
\rank(\tmr\mt)\geq \rank(\tmr)+D-N \geq \tilde{L}-(N-D).
\label{eqn:cor1}
\end{equation}
The proof follows from (\ref{eqn:cor1}) and Theorem \ref{thm1}.
\end{proof}
 
\vspace{-3mm}
\section{Construction of Subspace-Aware Index Codes}
It is well-known that the optimization problem in \eqref{eqn:optim} is 
NP-hard. In order to solve (\ref{eqn:optim}), we make a practical assumption that 
the users can tolerate a decoding error of at most $\epsilon$. That is,

\vspace{-4mm}
{\small
\begin{equation}
\Big\|\md_j\begin{bmatrix}\ms_j\vx\\ \vy\end{bmatrix} - \vx_{\mathcal{R}_j}\Big\|\leq\epsilon, \, \forall j.
\end{equation}}
Note that, subspace-unaware case with USI can be seen as special cases of~\eqref{eqn:optim}. 
Index codes over real field for this case has been studied previously in the literature~\cite{iitc}. 
It is known that a subspace-unaware linear index code matrix can be obtained by solving a 
matrix completion problem~\cite{dimakis,iitc}. However, the optimization problem in 
\eqref{eqn:optim} is more challenging compared to the conventional matrix completion problems. 
This is due to the fact that an indeterminate element in $\ma$ affects multiple entries 
in the resultant $\tmr\mt$ matrix in~\eqref{eqn:optim}, which is not the case in 
conventional matrix completion problems. In the next subsection, we consider the 
design of subspace-aware/unaware index codes with CSI/USI in a 
unified optimization framework.

\vspace{-3mm}
\subsection{Construction Algorithm for Index Codes}
\vspace{-0mm}
Let $\mathbf{Z}\triangleq[\mathbf{Z}_1^T,\cdots,\mathbf{Z}_U^T]^T$ be a rank $r$ matrix
and ${\bf Z}_j\in\mathbb{R}^{V_j\times D}$. Now, the optimization problem can be formulated as

\vspace{-6mm}
\begin{equation}
\label{opt1}\hspace{-1mm}
\min_{\{\mathbf{Z}_j,\mathbf{A}_j\}_{j=1}^U} \sum\limits_{j=1}^{U}\|\mathbf{Z}_j-(\mr_j-\ma_j\ms_j)\mt\|_F^2
=\hspace{-1mm}
\min_{\mathbf{Z},\{\mathbf{A}_j\}_{j=1}^U} \|\mathbf{Z}-\widetilde{\mr}\mt\|_F^2.
\end{equation}

\vspace{-1mm}
We solve the optimization problem in \eqref{opt1} for a range of values of $r$ and choose the minimum value of $r$ for which the optimization was feasible 
(i.e., all the constraints were satisfied) as the length of the index code ($L^*$).

\vspace{-2mm}
\begin{algorithm} 
\small                      
\caption{Subspace-aware Index Code Design}         
\label{altmin}                           
\begin{algorithmic} [1]
\REQUIRE $\{\mathbf{R}_j\}_{j=1}^U,\{\mathbf{S}_j\}_{j=1}^U,\mathbf{T},\epsilon,t_{\text{max}}$
\STATE \textbf{Initialize:} Random initial values\\ 
$t=0;\mathbf{X}\leftarrow \mathbf{X}(0)$; $\mathbf{Y}\leftarrow \mathbf{Y}(0)$; $\mathbf{A}_j=\mathbf{A}_j(0),\;\forall j$ 
\STATE \textbf{Update} $\mathbf{\{A_j\}}_{j=1}^U$\\ 
\textit{Solve:}  $\min_{\{\mathbf{A}_j\}_{j=1}^U} \sum\limits_{j=1}^{U}\|\mathbf{Z}_j(t)-(\mr_j-\ma_j\ms_j)\mt\|_F^2$  \\
\quad\qquad\qquad$\mathbf{A}_j(t+1) = (\mr_j\mt-\mathbf{Z}_j(t))(\ms_j\mt)^\dagger$
\STATE \textbf{Update} $\mathbf{X}$\\ 
\textit{Solve:}  $\min_{\mathbf{X}}\|\mathbf{XY(t)}-\widetilde{\mr}(t+1)\mt\|_F^2$  \\
\quad\qquad\qquad$\mathbf{X}(t+1) = \widetilde{\mr}(t+1)\mt \mathbf{Y}(t)^\dagger$
\STATE \textbf{Update} $\mathbf{Y}$\\ 
\textit{Solve:}  $\min_{\mathbf{Y}}\|\mathbf{X(t+1)Y}-\widetilde{\mr}(t+1)\mt\|_F^2$  \\
\quad\qquad\qquad$\mathbf{Y}(t+1) = \mathbf{X}(t+1)^\dagger\widetilde{\mr}(t+1)\mt $
\IF{$\|\mathbf{Z}(t)-\widetilde{\mr}(t)\mt\|_F\leq \epsilon$ or $t=t_{\text{max}}$}
\STATE \textbf{return} $\{\mathbf{A}_j(t+1)\}_{j=1}^U$
\ELSE
\STATE $t\leftarrow t+1$
\STATE \textbf{return to} \text{Step} $2$
\ENDIF
\end{algorithmic}
\end{algorithm}

\vspace{-2mm}
We factorize $\bf Z$ as $\mathbf{Z}=\mathbf{X}\mathbf{Y}$, where 
$\mathbf{X} \in \mathbb{R}^{(\sum_jV_j)\times r}$, and 
$\mathbf{Y} \in \mathbb{R}^{r\times D}$. 
The optimization problem in \eqref{opt1} is not convex in $\mathbf{X}$, $\mathbf{Y}$ 
and $\mathbf{A}$ simultaneously; however, it is convex in $\mathbf{X}$ (or $\mathbf{Y}$ 
or $\mathbf{A}$) when the rest of the optimization variables are fixed. In fact, here,
each of the sub-problems can be solved in a closed form. Note that, $\bf Z$ is a 
rank $r$ approximation of $\tmr\mt$ (with an error of $\epsilon$,
i.e., $\|\mz-\tmr\mt\|_F\leq \epsilon$; index codes over $\mathbb R$ enable us to obtain 
such a rank $r$ approximation). The steps in solving this optimization 
problem are listed in Algorithm \ref{altmin}. The alternating minimization method 
is guaranteed to converge to a locally optimum solution for a sufficiently large 
number of iterations \cite{Fazel}.

Let $\widetilde{\mathbf{Z}}$ be the matrix formed by choosing the $L^*$ linearly 
independent rows of ${\bf Z}$. Now, we set $\mc\mt=\widetilde{\mathbf{Z}}$. 
At every transmission instant, if the low-dimensional vector $\vw$ is
available at the DS, then the matrix $\mc\mt$ can be used for index coding to 
generate $\vy=\mc\mt\vw$, else the matrix $\mc\mt\mt^{\dagger}$ is used (since
$\mc\mt\mt^{\dagger}\vx=\mc\mt\mt^{\dagger}\mt\vw=\mc\mt\vw=\mc\vx=\vy$).


\vspace{-2mm}
\subsection{Decoding Error Analysis}
\begin{theorem}
For an index code constructed using the proposed algorithm such that 
$\|\mz-\tmr\mt\|_F\leq \epsilon$, the decoding error is bounded above by $\epsilon$. 
\end{theorem}
\hspace{-3.5mm}{\em Proof.}
Let $\mr\hat{\vx}$ be the vector decoded at the receivers. Then, the decoding error is 
\begin{equation}
\|\mr\vx-\mr\hat{\vx}\|=\|\mr\mt\vw-(\ma\ms\mt\vw+\mb\mc\mt\vw)\|,
\label{eqn:derr}
\end{equation} 
where the values of the matrices $\ma$ and $\mc\mt($ $=\widetilde{\bf Z})$ are obtained 
from Algorithm~\ref{altmin}. We choose $\mb$ such that ${\bf Z=B\widetilde{Z}}$. 
This is possible due to the following reason. Without loss of generality, we can express 
$\mz$ as $\mz=[\widetilde{\mz}^T,\, \bar{\mz}^T]^T$, where $\bar{\mz}$ is the matrix of 
$\sum_jV_j-L$ linearly dependent rows of $\mz$. Therefore, the rows of $\bar{\mz}$
are in the $\spn($rows of $\widetilde{\mz})$, i.e., $\bar{\mz}=\mg\widetilde{\mz}$
for some matrix $\mg\in\mathbb{R}^{\sum_jV_j-L\times L}$. Hence, by choosing $\mb$ as
$\mb=[\mi_L,\, \mg^T]^T$, we have $\mb\widetilde{\mz}=\mz$. Further, without loss of
generality, we assume $\|\vw\|_2\leq 1$.

Now, from \eqref{eqn:derr}, the decoding error can be bounded as
\[
\|\mr\vx-\mr\hat{\vx}\|\leq\|\tmr\mt-\mb\widetilde{\mathbf{Z}}\|_F\|\vw\|_2\leq\epsilon.
\quad\quad\quad\qedsymbol\]
{\em Remark 1}: The matrices $\widetilde{\mz}$ and $\bar{\mz}$ can be easily obtained from 
$\mz$ using one of the many commonly known techniques such as using QR decomposition, and 
$\mg=\bar{\mz}\widetilde{\mz}^\dagger$. \\
{\em Remark 2}: In the proposed decoding strategy, the users need not be aware of the
subspace matrix $\mt$ for decoding. Using the matrix $\md_j$, each user can directly 
decode $\vx_{\mathcal{R}_j}$.

\vspace{-3mm}
\section{Numerical Results}
\label{sec5}
\vspace{-1mm}
Here, we present numerical results for the proposed index code construction 
algorithm and analyze its performance.
\vspace{-4mm}
\subsection{Comparison}
\label{sec5a}
\vspace{-1mm} {
First, we consider the index coding problem with USI from~\cite{iitc}\footnote{The 
algorithm in~\cite{iitc} can solve only the subspace-unaware index coding problem with 
USI (conventional matrix completion problem) which is a special case of the problem we 
consider in this paper.}, where $U=4$, $\mathcal{R}_i=\{i\}$ for $i=1,2,3,4$, and

\vspace{-2mm}
{\small
\[\ms_1=\begin{bmatrix}0&1&0&0\\ 0&0&1&0\end{bmatrix},\quad
\ms_2=\begin{bmatrix}1&0&0&0\\ 0&0&1&0\end{bmatrix},\]
\[\ms_3=\begin{bmatrix}0&1&0&0\\ 0&0&0&1\end{bmatrix},\quad
\ms_4=\begin{bmatrix}1&0&0&0\end{bmatrix}.\]}

\vspace{-4mm}
Using our proposed algorithm, the $\tmr$ matrix obtained for $\epsilon=10^{-10}$ is 

\vspace{-3mm}
{\small
\[ \begin{bmatrix}1& -0.9122099& -1.0733032& 0\\
-1.0962388& 1& 1.1765966& 0\\
0& 0.8499089& 1& -1.0952999\\
-0.8506375& 0& 0& 1
\end{bmatrix}.\]}
The rank of this matrix is $2$, which is the optimal index code length for this 
problem as given in~\cite{iitc}. We obtain the index code for this problem by choosing 
the linearly independent rows in the above matrix. The index code thus obtained is

\vspace{-2mm}
{\small
\[ \begin{bmatrix}1& -0.9122099& -1.0733032& 0\\
-0.8506375& 0& 0& 1
\end{bmatrix}.\]}
For example, when $\vx=[1, 1, -1, 2]^T$ is the source-data, the reconstructed values 
at the users, from the designed index code were $[1, 0.9999999, -1, 2]$; this gives a 
decoding error of $\|\vx-\hat{\vx}\|=4.02\times 10^{-14}$. For the same problem, consider 
the source-data to be low-dimensional with 
{\scriptsize
$\mt=\begin{bmatrix}1& -2& 1& 1\\
1& 1& -1& 2
\end{bmatrix}^T$}. Now, $\vx=[1, 1, -1, 2]^T=\mt[0, 1]^T$. For this linear subspace,
we get a subspace-aware index code of length 1 given by the matrix 
$[0, -0.3936923, 0.2644723, -0.1412844]$, and a corresponding decoding error of 
$8.34\times 10^{-14}$ at the users.

\begin{figure}
\centering
\includegraphics[width=3.5in, height=2.5in]{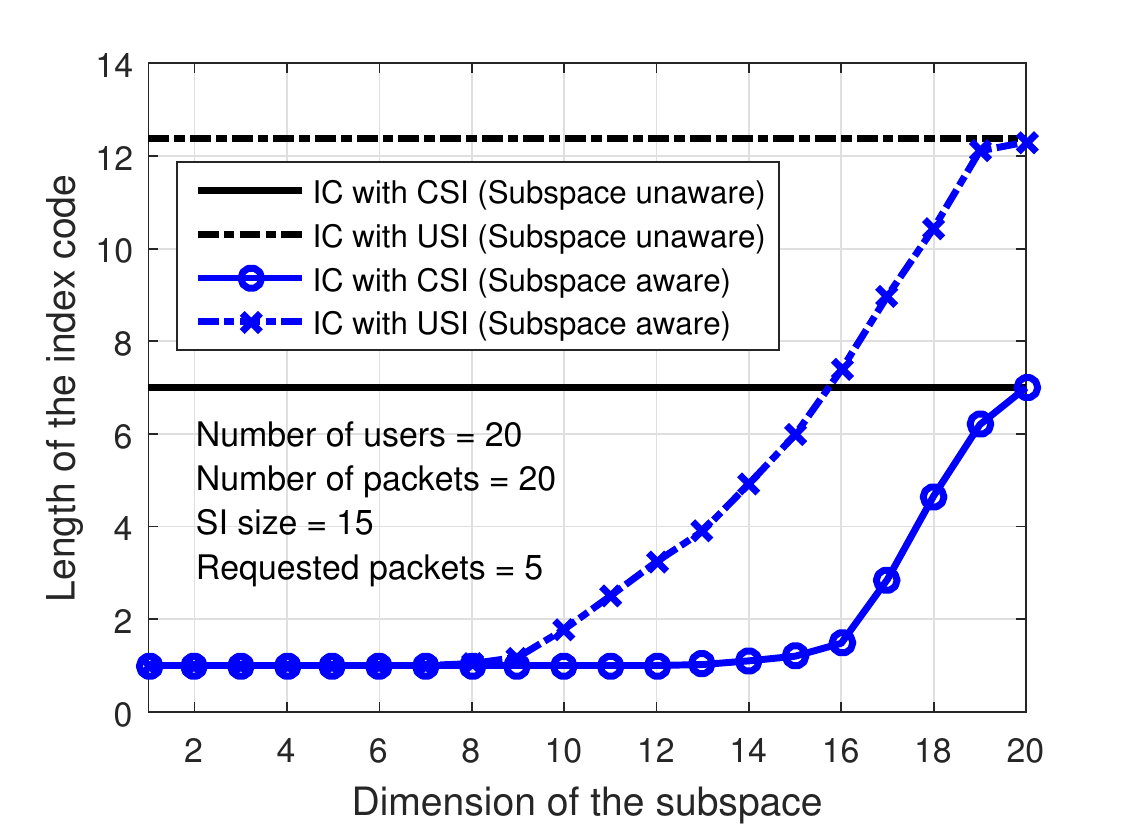} 
\caption{Index code length obtained using the proposed algorithm for different subspace dimensions.}
\label{fig:sim1}
\end{figure}

\begin{figure*}
    \centering
    \subfigure[]{
        \includegraphics[width=0.45\textwidth]{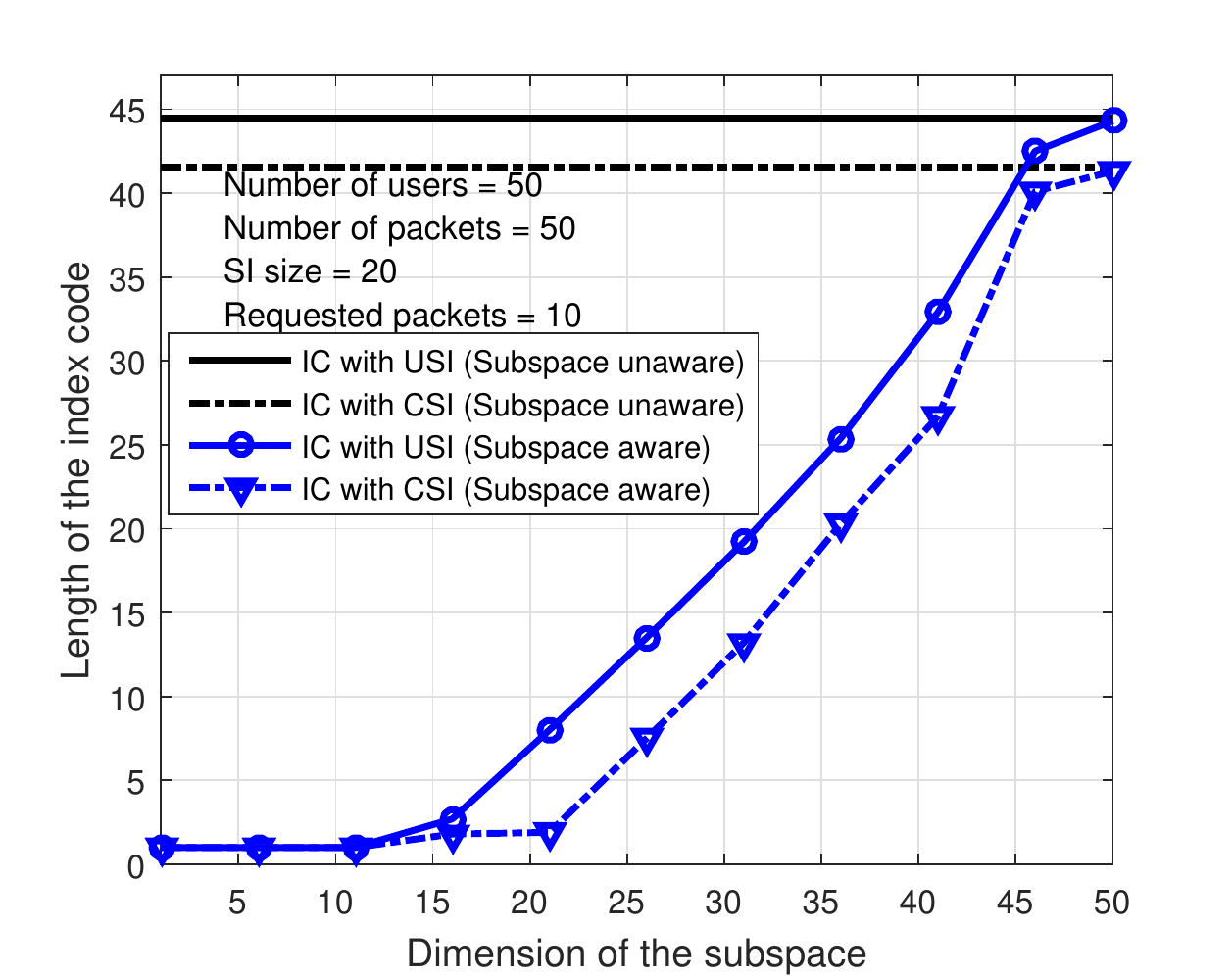}
        } \hspace{4mm}
    \subfigure[]{
        \includegraphics[width=0.45\textwidth]{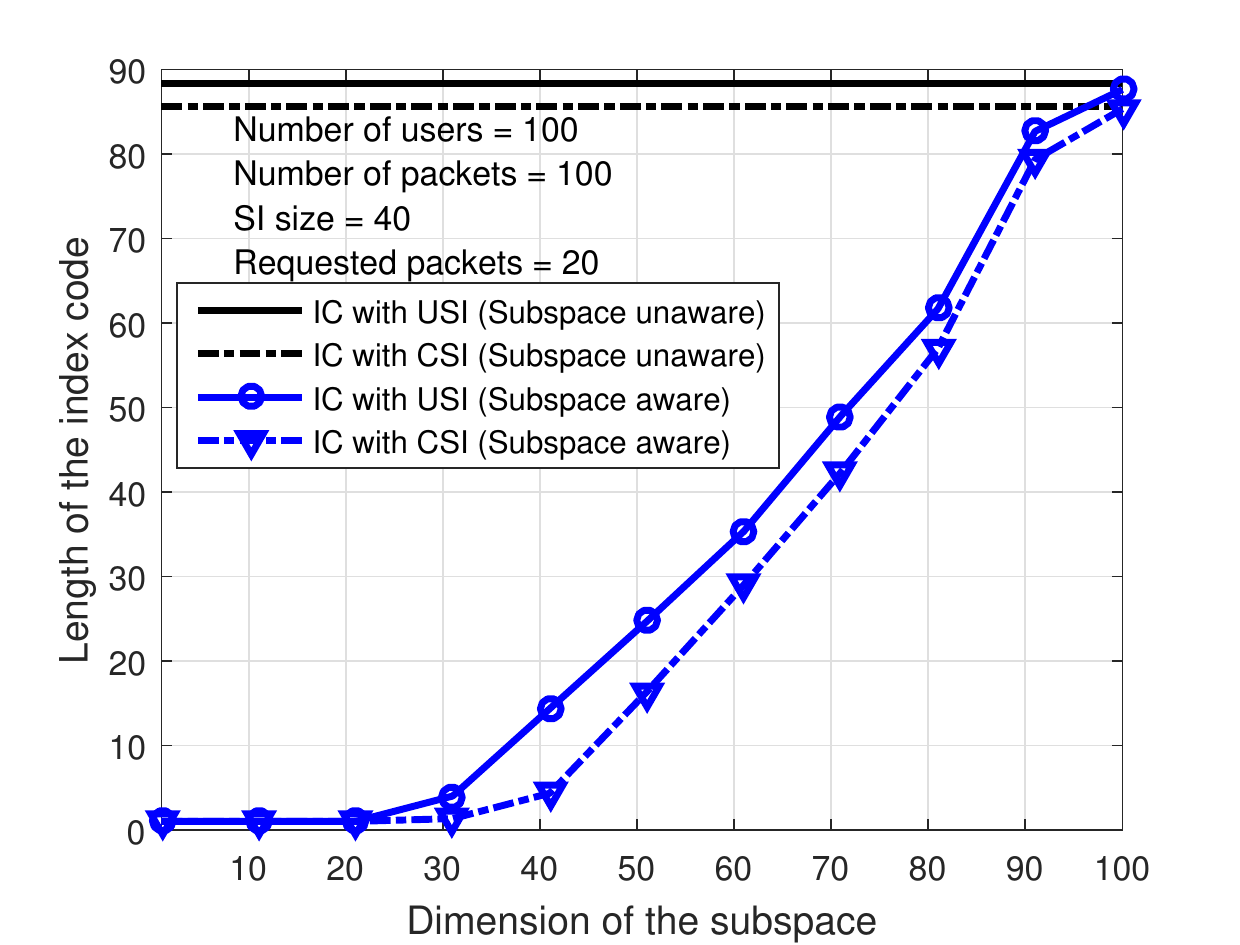}
        }        
    \caption{Index code length obtained using the proposed algorithm for different system parameters.}
    \label{fig:sim3}
\end{figure*}

\begin{figure}
\centering
\includegraphics[width=3.5in, height=2.5in]{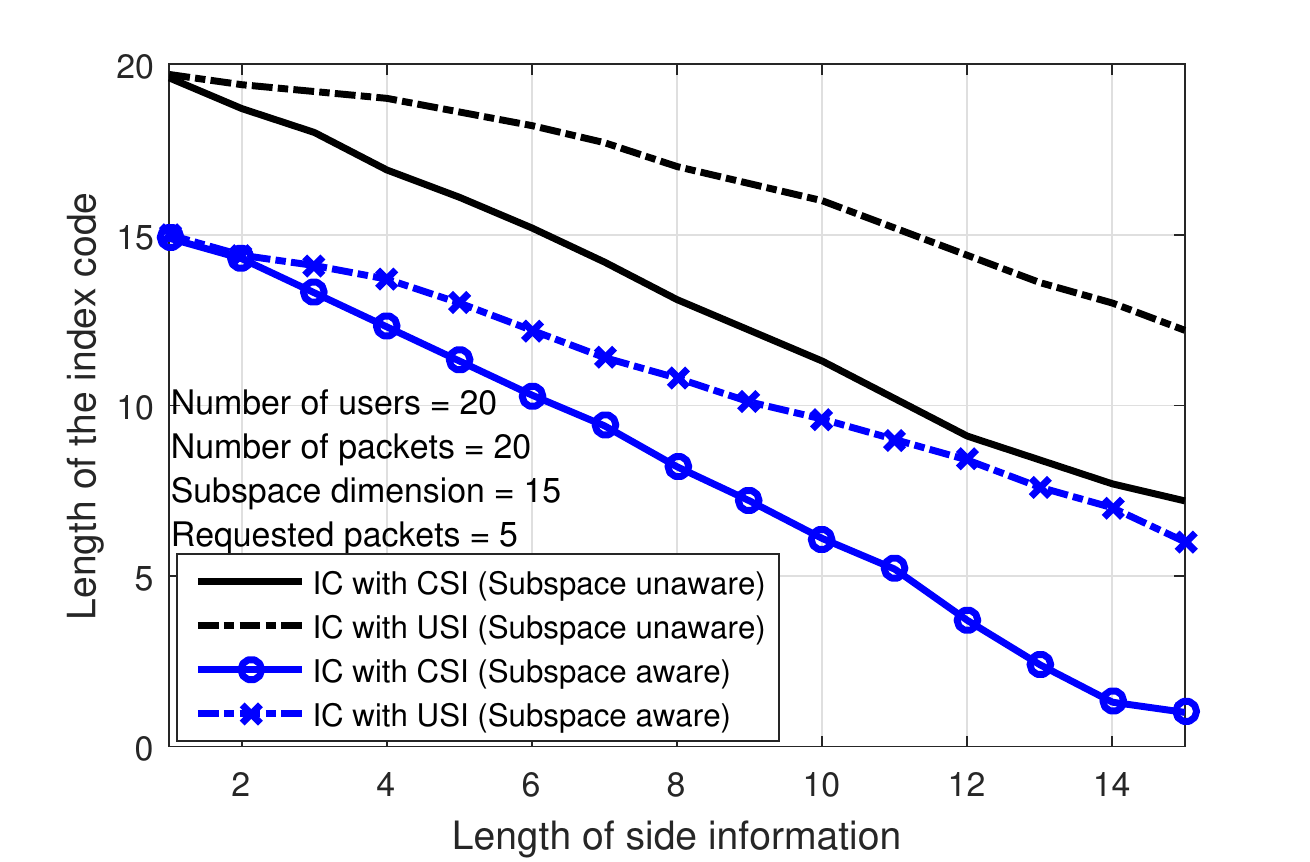} 
\caption{Index code length obtained using the proposed algorithm for different side information sizes.}
\label{fig:sim2}
\end{figure}

\vspace{-4mm}
\subsection{Simulation results}
\vspace{-1mm}
We simulated a simple multicast video-streaming scenario with $N=20$ source-data packets 
and $U=20$ users requesting $|\mathcal{R}_j|=5$ data packets, each. We evaluated the  
index code length averaged over several instances for four different cases -- namely, $(1)$ DS is 
subspace-unaware and the
SI is uncoded, $(2)$ DS is subspace-unaware and the SI is coded, $(3)$ DS is subspace-aware 
and SI is uncoded, and $(4)$ DS is subspace-aware and the SI is coded.
For fair comparison, we consider the same requirement matrix $\mr$ for all the cases.

In Figure \ref{fig:sim1}, we plot the average index code length obtained using our proposed
algorithm for different subspace dimensions fixing the SI length at each user to be 
$M_j=15$. We see that when the source-data is low-dimensional, the average index code 
lengths obtained for the subspace-aware cases are significantly less than that of the 
subspace-unaware cases. The average index code lengths for the subspace-unaware cases 
are 7 (CSI) and 12.4 (USI). Whereas, in subspace-aware case, for $D<9$, the average
index code length is 1.1. Therefore, subspace-aware index codes reduce the transmissions
required by about 91\% for the USI case and by about 85\% for the CSI case\footnote{ As 
the number of packets and users increase, the difference between the average index code length 
for the CSI and that of the USI case decreases.}. Also, for $9\leq D<20$, we observe that
the subspace-aware index codes consistently outperform the subspace-unaware index codes 
by considerable margin.

Furthermore, from Fig. \ref{fig:sim3}, we can see that even when the number of packets are $50$ or $100$, the
subspace-aware index code outperforms the subspace-unaware index codes. For example, when the
subspace dimension is half that of the number of packets (i.e., $D=25$ when $N=50$, and $D=50$
when $N=100$), subspace-aware index codes have code lengths that are $70$\% smaller compared to
that of subspace-unaware index codes for uncoded side information, and subspace-aware index 
codes have $81$\% advantage over subspace-unaware index codes for coded side information. Thus,
we can observe that irrespective of the number of packets, the subspace-aware index codes
provide significant throughput gains over the subspace-unaware index codes.

In Figure \ref{fig:sim2}, we evaluate the performance of the proposed algorithm for varying 
SI lengths ($M_j$) and fixing the subspace dimension at $D=15$. As before, we can see that 
the subspace-aware index codes have significant throughput gains over the subspace-unaware 
index codes in both the USI and CSI cases. For instance, when $M_j=10$, the subspace-aware 
cases (both USI and CSI) have an average index code length that is at least 30\% smaller than 
that of the subspace-unaware cases.
}

\vspace{-1.5mm}
\section{Conclusion}
\vspace{-1mm}
In this paper, we studied a generalization of the index coding problem that exploits 
source-data's structure to improve the system-throughput. We analytically characterized 
the length of the subspace-aware index codes and proposed an algorithm to obtain near 
optimal index codes. We showed that this approach significantly outperforms the conventional 
approaches when the source-data belongs to a low-dimensional subspace. 
{Index coding for the case when the source-data belongs to a non-linear subspace 
or manifold is an interesting direction for future research. Further, network codes can also
be constructed using the proposed algorithm, once the network coding problem is converted to an
equivalent index coding problem \cite{Effros}}.


\end{document}